\PassOptionsToPackage{table,dvipsnames}{xcolor}
\documentclass[sigconf,natbib=true,screen=true,anonymous=false,review=false]{acmart}

\usepackage{acronym}
\usepackage[noend]{algorithmic}
\usepackage{algorithm,balance}
\usepackage{bm}
\usepackage[T1]{fontenc}
\usepackage{graphicx}
\usepackage{hyperref}
\usepackage{listings}
\usepackage{multirow}
\usepackage{natbib}
\usepackage{subcaption}
\usepackage{xcolor}
\usepackage[inline]{enumitem}
\usepackage[toc,page]{appendix}
\usepackage{centernot}
\usepackage{amsthm}
\usepackage{mathtools}
\usepackage{dsfont}
\usepackage{pifont}
\usepackage{rotating}
\usepackage{stmaryrd}

\copyrightyear{2025}
\acmYear{2025}
\setcopyright{cc}
\setcctype{by}
\acmConference[SIGIR '25]{Proceedings of the 48th International ACM SIGIR
Conference on Research and Development in Information Retrieval}{July 13--18,
2025}{Padua, Italy}
\acmBooktitle{Proceedings of the 48th International ACM SIGIR Conference on
Research and Development in Information Retrieval (SIGIR '25), July 13--18,
2025, Padua, Italy}\acmDOI{10.1145/3726302.3730020}
\acmISBN{979-8-4007-1592-1/2025/07}

\usepackage{etoolbox}

\definecolor{todo}{RGB}{200, 0, 0}
\definecolor{nk}{RGB}{0, 100, 0}
\definecolor{ho}{RGB}{0, 50, 150}
\definecolor{black}{gray}{0}
\definecolor{white}{gray}{1}

\acrodef{IR}{information retrieval}
\acrodef{LTR}{learning to rank}
\acrodef{RecSys}{recommender systems}
\acrodef{CTR}{click-through rate}
\acrodef{MF}{matrix factorization}
\acrodef{LBD}{learned beta distributions}
\acrodef{CMF}{Confidence-Aware Matrix Factorization}
\acrodef{MF}{Matrix Factorization}
\acrodef{PDF}{probability density function}
\acrodef{PMF}{probability mass function}
\acrodef{CDF}{cumulative distribution function}
\acrodef{RMSE}{root mean squared error}
\acrodef{MAE}{mean absolute error}
\acrodef{NDCG}{non discounted cumulative gain}
\acrodef{BPMF}{Bayesian Probabilistic Matrix Factorization}
\acrodef{LBD-S}{LBD Static}
\acrodef{LBD-A}{LBD Adaptive}
\acrodef{PRP}{probability ranking principle}
\acrodef{LTR}{learning to rank}
\acrodef{SERP}{search engine result page}
\acrodef{PL}{Plackett-Luce}
\acrodef{VLPL}{variable length Plackett-Luce}
\acrodef{EA}{expected attractiveness}
\acrodef{MLP}{multilayer perceptron}

\theoremstyle{definition}
\newtheorem{theorem}{Theorem}[section] %

\newcommand{\E}[1]{\mathbb{E}_{#1}}

\newcommand{\dm}[0]{\frac{\delta}{\delta m}}
\newcommand{\one}[1]{\mathds{1}{\Big[#1\Big]}}

\newcommand{\Rq}[0]{\mathcal{R}(q)}

\newcommand{\rr}[0]{r}

\newcommand{\rel}[0]{\rho}
\newcommand{\pt}[0]{\phantom{^\triangledown}}
\newcommand{\shortleftarrowtwo}[0]{\mkern-4mu\shortleftarrow\mkern-4mu}

\allowdisplaybreaks
\author{Norman Knyazev}
\affiliation{%
	\institution{Radboud University}
	\city{Nijmegen}
	\country{The Netherlands}
}
\email{norman.knyazev@ru.nl}
\author{Harrie Oosterhuis}
\affiliation{%
	\institution{Radboud University}
	\city{Nijmegen}
	\country{The Netherlands}
}
\email{harrie.oosterhuis@ru.nl}

\acmSubmissionID{232}

\title{Learning to Rank with Variable Result Presentation Lengths}

\begin{document}

\ccsdesc[500]{Information systems~Learning to rank}
\ccsdesc[500]{Information systems~Presentation of retrieval results}
\keywords{Learning to Rank, Variable Length Presentations, Policy Learning}

\begin{abstract}
\Ac{LTR} methods generally assume that each document in a top-$K$ ranking is presented in an equal format.
However, previous work has shown that users' perceptions of relevance can be changed by varying presentations, i.e., allocating more vertical space to some documents to provide additional textual or image information.
Furthermore, presentation length can also redirect attention, as users are more likely to notice longer presentations when scrolling through results.
Deciding on the document presentation lengths in a fixed vertical space ranking is an important problem that has not been addressed by existing \ac{LTR} methods.

We address this gap by introducing the variable presentation length ranking task, where simultaneously the ordering of documents and their presentation length is decided. 
Despite being a generalization of standard ranking, we show that this setting brings significant new challenges:
Firstly, the probability ranking principle no longer applies to this setting, and secondly, the problem cannot be divided into separate ordering and length selection tasks.

We therefore propose VLPL -- a new family of Plackett-Luce list-wise gradient estimation methods for the joint optimization of document ordering and lengths.
Our semi-synthetic experiments show that VLPL can effectively balance the expected exposure and attractiveness of all documents, achieving the best performance across different ranking settings.
Furthermore, we observe that even simple length-aware methods can achieve significant performance improvements over fixed-length models.
Altogether, our theoretical and empirical results highlight the importance and difficulties of combining document presentation with \ac{LTR}.
\end{abstract}

\maketitle
\acresetall

\section{Introduction}

Ranking models generally aim to identify the most relevant documents from a set of candidates and present them to the user in the order of decreasing relevance on a \ac{SERP}~\citep{liu2009learning}.
However, the presentation of individual documents within a \ac{SERP} is rarely considered when choosing the document ordering, with the documents usually simply modeled as fixed units of information or relevance.
Nevertheless, in many cases the platforms can choose how each document is presented to the user, which in turn affects how the user perceives and interacts with those documents. 
We call this setting where both the order of the documents as well as their presentation lengths are chosen as the variable document presentation length setting -- see Figure~\ref{fig:variable_length_ranking_example} for an example.
This setting is characterized by the trade-off between presenting more information for displayed documents -- potentially making them more attractive -- or displaying more documents on the \ac{SERP} -- potentially including more relevant results.

\begin{figure}[tb]
\centering
    \includegraphics[width=\columnwidth]{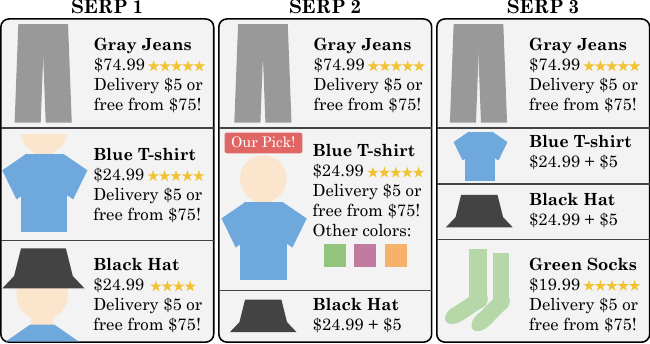}
\vspace{-7.3mm}
\caption{Rankings with variable presentation length. Compared to SERP 1, SERP 2 presents the second item with more vertical space (length), allowing more information to be displayed.
In SERP 3, the second and third item are both presented with less vertical space, resulting in less information for each, but allowing an additional item to be displayed.
}%
\label{fig:variable_length_ranking_example}
\vspace{-1.4\baselineskip}
\end{figure}

Earlier research has shown that increasing the size of individual \ac{SERP} results to include additional textual information may make the user more likely to click through to the full document \citep{maxwell2017Study}, 
whilst shorter results may discourage the users from clicking \citep{kohavi2015Online}. 
Similarly, in e-commerce, adding images and product details has also been found to increase the engagement with those items \citep{bland2005Determinants,you2023What}.

Making a search result larger may reduce the probability that a user will miss it when browsing \citep{cutrell2007What}, and thereby, alleviate some position bias on that result~\citep{marcos2015Effect}.
However, as noted by \citeauthor{pathak2024Influence}~\citep{pathak2024Ranking,pathak2024Influence}, increasing the vertical size of some results generally means that fewer results fit on the \ac{SERP}~\citep{kohavi2015Online}.
In general, increasing the presentation size of an individual retrieved document also means that the documents below it are pushed down by the same amount. 
This, in turn, may make them less likely to be observed by the user and lead to an increase in the effects of position bias~\citep{kelly2015How,kim2016Pagination,craswell2008Experimental}.
Hence this setting requires a balancing of document ordering and their presentation lengths. 
To the best of our knowledge, no existing method can simultaneously optimize a ranking and variable presentation lengths.%

In this work, we formally introduce and study the variable document presentation length setting.
Our analysis reveals challenges that are not present in standard ranking:
most importantly, the \ac{PRP} does not apply, and as a result, the optimal placement strategy cannot be separated into independent ranking and length-selection strategies.
In order to tackle this problem, we extend the \ac{PL} ranking approach to optimize both the order and presentation lengths of documents.
Our novel model, the \ac{VLPL} distribution $\pi_{\text{VLPL}}$, represents each choice as a document placement paired with a specific length.
We build on the PL-Rank family of gradient estimators \citep{oosterhuis2021computationally,oosterhuis2022LearningtoRank} to introduce the novel VLPL$-1$ and VLPL$-2$ estimators.
These efficiently estimate the gradient of $\pi_{\text{VLPL}}$, with VLPL$-2$ having a higher sample-efficiency.
Moreover, as we show that the gradient estimation can be leveraged during model training as either in-processing or post-processing \citep{caton2024Fairness}, in total we propose four distinct methods for jointly optimizing ordering and length.
Our semi-synthetic experiments demonstrate that our approaches achieve a large and consistent improvement over baseline methods.
Overall, our findings underscore the need for ranking optimization methods with a more holistic view of result presentation that goes beyond the ``\emph{ten blue links}", as some links may look very different from others, and perhaps there should not be ten.

\section{Related Work}
\label{sec:related_work}

\Ac{LTR} models aim to identify the most relevant documents and present them to the user, most commonly in the form of a ranked list.
This is usually done by scoring each document and then sorting the documents based on their scores.
According to the \ac{PRP}~\citep{robertson1977Probability}, a scoring function should reflect the probability of each document being relevant to the user to produce optimal rankings.
In recent years, the Plackett-Luce family of list-wise methods \citep{plackett1975analysis} have gained popularity, owing to their simplicity, robustness and interpretability \citep{oosterhuis2021computationally,oosterhuis2022LearningtoRank,buchholz2022Lowvariancea}.
Importantly, however, these and other widely used ranking models \citep{fuhr1989Optimum,joachims2002optimizing,burges2010ranknet,wang2018lambdaloss,cao2007learning, xia2008listwise,oosterhuis2021computationally,oosterhuis2022LearningtoRank,buchholz2022Lowvariancea} only view ranking from the lens of document ordering, not concerning themselves with any further choices about document presentation.

A large body of work has shown that result presentation may significantly affect how users perceive and interact with each document.
\citet{maxwell2017Study} observed that users are more likely to interact with a result on a \ac{SERP} if it is presented with a longer snippet.
This is also consistent with the observations of \citet{kohavi2015Online}, who found that 
allocating the same amount of screen space to fewer but larger and more interactive ads led to significant increases in user engagement and conversions.
\citet{marcos2015Effect} reported that lower-ranking results with larger and more interactive snippets are noticed earlier and for longer compared to other snippets, with snippets containing images being particularly salient.
\citet{you2023What} also found that the image of a \ac{SERP} result may have a drastic impact on how likely the user is to find the document attractive~\citep{aiello2016Role,chuklin2015click,yue2010Position}. 

On the other hand, the number of documents visible to the user at a given time also significantly affects user interactions \citep{oosterhuis2020topkrankings,kim2016Pagination,joachims2002optimizing,Schnabel2016,kelly2015How,craswell2008Experimental,fang2019intervention,sweeney2006Effective}.
\citet{kelly2015How} showed that having more results on each \ac{SERP} leads to users interacting with more documents, 
whilst \citet{kim2016Pagination} found that the need to scroll, in contrast with pagination, significantly reduces the speed at which the users are able to browse the result list.
Taken together, the above findings suggest that decreasing snippet length, which allows to include more results on a single \ac{SERP}, may lead to a higher number of interactions with the full ranking.
Furthermore, the above results are also consistent with the findings of \citet{joachims2002optimizing} on position bias~\citep{craswell2008Experimental,cutrell2007What,lorigo2008Eye}, who saw that the probability that the user observes a document decreases with the position of that document, with particularly significant drops after the end of each page.

Our work is most related to that of~\citet{mao2024Wholea} and \citeauthor{pathak2024Influence}~\citep{pathak2024Ranking,pathak2024Influence}. \citeauthor{mao2024Wholea} aim to learn the causal graph based on the features of the \ac{SERP}, which also includes document rank and document height. This causal graph is then used to correct for the impact of those variables on the click data in order to learn the intrinsic document relevance.
However, the authors assume a fixed length for each document and do not learn how to the document lengths interact across different documents.
\citet{pathak2024Ranking}, on the other hand, focused on the setting where snippets can either contain the headline, summary, image, or any combination of the above, with each element taking up additional space on the screen.
They find that changing the presentation of a document may also change where it should be placed in the ranking, particularly due to the increase in time the users spend examining the longer presentations.
Nevertheless, the authors focus on the setting where the document lengths are chosen prior to determining the document order, whilst in Section~\ref{sub:PRP} we prove that this can lead to suboptimal solutions.
Finally, their approach assumes that user's examination behavior can be measured separately for each document, and does not provide a ranking function for previously unseen queries.

\section{Background: The Standard Ranking Setting}

The core \ac{LTR} task is to find a ranking $y = [d_1, d_2, \ldots, d_{K}]$ of length $K$ that orders documents by decreasing relevance $\rel_i$ for a given query $q$ and document set $d\in D$. %
To assess the quality of a ranking, per-position discounting factors $\theta(i)$ are commonly used, e.g., one-over-rank and DCG weighting.
Stochastic ranking models \citep{oosterhuis2021computationally,ustimenko2020StochasticRank} can optimize a policy $\pi$ that provides a probability distribution over rankings to maximize an objective of the form:
\begin{equation}
  \label{eq:standard_objective}
  \mathcal{R}_{\text{LTR}}(q) = \E{\pi(y)}\Bigg[
    \sum_{i=1}^{K} \theta(i) \rel(d_i)
  \Bigg] = \sum_{y} \pi(y) \sum_{i=1}^{K} \theta(i) \rel(d_i).
\end{equation}
In click modeling \citep{chuklin2015click}, the reward $\rel$ may often represent the perceived relevance or \emph{attractiveness} of the document, i.e., the probability the user will click on the document when exposed to it.

\noindent
The policy of a Plackett-Luce \ac{LTR} model \citep{plackett1975analysis, oosterhuis2021computationally} is defined as:
\begin{equation}
  \label{eq:PL}
   \pi_{\text{PL}}(y) = \prod_{i=1}^{K}\pi_{\text{PL}}(y_i \mid y_{1:i-1}) = \prod_{i=1}^{K} \frac{
    e^{m(d_i)} \mathds{1}[\substack{d_i \not\in y_{1:i-1}}]
  }{
    \sum_{d'=1}^{|D|}
    e^{m(d')} \mathds{1}[\substack{d' \not\in y_{1:i-1}}]},
\end{equation}
where $m(d)$ is the score for document $d$ under some scoring function $m$.
In the Plackett-Luce distribution, documents are sequentially sampled from a softmax distribution without replacement. %
Furthermore, the sampling probabilities in each position only depend on the scores of unplaced documents.
Whilst the optimization of $\pi_{\text{PL}}$ policies can be done with standard policy gradient methods (i.e., via REINFORCE~\citep{williams1992Simple}), \citet{oosterhuis2021computationally, oosterhuis2022LearningtoRank} has proposed the PL-Rank family of methods, specialized to efficiently estimate gradients of Plackett-Luce distributions for \ac{LTR} optimization.

Nevertheless, the standard ranking setting is not applicable if documents can be presented with different lengths.
In particular, this is mainly due to two limitations of its objective (Equation~\ref{eq:standard_objective}):
\begin{enumerate*}
\item the rank weights $\theta(i)$ only consider the ranks of documents and not their presentation lengths; and
\item a ranking always consists of $K$ placed documents.
\end{enumerate*}
However, as discussed in Section~\ref{sec:related_work}, the length of a document can affect the attractiveness and the amount of attention it receives from users~\citep{marcos2015Effect,maxwell2017Study}.
Therefore, $\theta(i)$ should adapt according to the attention that is expected from users~\citep{jeunen2024Normalised} and, similarly, attractiveness $\rel(d_i)$ should also change with length.
Furthermore, changing document presentation length should affect the number of documents that can be presented on a single \ac{SERP}~\citep{pathak2024Ranking}.
None of these aspects are captured by the standard ranking setting.

\section{Variable Presentation Length Problem Setting}
\label{sec:problem_setting}

This section introduces a formalization of the variable presentation length ranking setting.
Subsequently, we also discuss some of its theoretical implications.

\subsection{Slot based ranking}

To capture the unique characteristics of this setting, 
we represent each placement as a choice of document $d$ and its length $l$, with a full ranking being $y = [(d_1, l_1),(d_2, l_2), \ldots, (d_{|y|}, l_{|y|})]$, where each $d$ can only be placed once.
Crucially, we assume that each ranking consists of $K$ positions or slots: $\text{len}(y) = \sum_{j=1}^{|y|} l_i = K$.
However, unlike the standard setting, each position now refers to a physical space in the ranking, not a moment of placement.
The lengths $l$ thus represent the number of slots each document occupies in the ranking.
We also denote the first slot $d_i$ occupies as $s_i = 1 + \sum_{j=1}^{i-1} l_i$.
The above changes then allow us to define a new optimization objective called the \ac{EA} of a ranking:
{\fontsize{8.3pt}{10.14pt}
\begin{align}
  \label{eq:opt_objective}
  \Rq = \E{y}\Bigg[
    \sum_{i=1}^{|y|} \theta({s_i},{l_i}) \rel({d_i},{l_i})
  \Bigg]
  = \!\!
  \sum_{y} \pi(y) \! \sum_{i=1}^{|y|} \theta({s_i},{l_i}) \rel({d_i}, {l_i}).
\end{align}}

An important consequence of the above differences with standard ranking is that the number of unique documents $|y|$ in a ranking can be lower than $K$. 
Our task is then to choose the document-length pairs that maximize the weighted attractiveness in Equation~\ref{eq:opt_objective} without exceeding the slot budget $K$ or duplicate document placements.
In contrast with standard ranking, both the position weight $\theta$ and the reward $\rel$ depend on the presentation length $l$.
For this work, we assume that $\rel$ is fixed for each $(d,l)$, i.e., for document $d$ each presentation length is associated with a single relevance value, which is the expected \ac{CTR} of $d$ when exposed to the user.
Finally, to make our analysis and approach tractable, we assume a maximum document length of $L$ slots.
Naturally, document lengths may in reality be adjustable even up to a single pixel or may instead be fixed (but distinct) for each document.
We consider both of these important next steps to be addressed by future work.

Finally, we note that our proposed setting has some overlap with that of \citet{pathak2024Ranking}: they also capture that varying document presentation lengths can lead to different number of items appearing above the fold. 
See Section~\ref{sec:related_work} for a further discussion.

\subsection{Theoretical analysis on optimal rankings}
\label{sub:PRP}

\begin{table*}[]
\caption{The impact of discount weights on optimal ordering. For three documents $A$, $B$, $C$ with respective relevances $1.0$, $0.6$ and $0$ (left, same for all lengths $l$), different choices of discount weights $\theta_1$ (center left) and $\theta_2$ (center right) with $K=L=3$ lead to different expected rewards for different rankings (right). $BAA$ represents ranking $y=[(B,1), (A,2)]$, whilst $AB(C)$ is equivalent to both $AB$ and $ABC$. Best ranking per $\theta$ is underlined. 
\vspace{-3.3mm}
}
\label{tab:example}
  \addtolength{\tabcolsep}{-0.2em}\hspace{3mm}\begin{tabular}{cccc}
    $d$ & $\rel(d)$\\
    \toprule
    A & 1.0\\
    B & 0.6\\
    C & 0.0 \\
    \bottomrule
   \end{tabular}
   \,
    \begin{tabular}{cccc}
    $\theta_1$ & l=1 & l=2 & l=3 \\
    \toprule
    s=1 & 0.500 & 0.667 & 0.750 \\
    s=2 & 0.333 & 0.500 & - \\
    s=3 & 0.250 & - & - \\
    \bottomrule
   \end{tabular}
   \,
    \begin{tabular}{cccc}
    $\theta_2$ & l=1 & l=2 & l=3 \\
    \toprule
    s=1 & 0.631 & 0.815 & 0.895 \\
    s=2 & 0.500 & 0.715 & - \\
    s=3 & 0.431 & - & - \\
    \bottomrule
   \end{tabular}
   \,
  \begin{tabular}{cccccccccc}
    & AAA & AAB & ABB & BAA & BBA & BBB & AA & AB(C) & BA\\
    \toprule
    $\theta_1$ & 0.750 & \underline{0.817} & 0.800 & 0.800 & 0.650 & 0.450 & 0.667 & 0.700 & 0.633\\
    $\theta_2$ & 0.895 & 1.074 & 1.060 & \underline{1.094} & 0.920 & 0.537 & 0.815 & 0.931 & 0.879\\
    \bottomrule
   \end{tabular}
  \vspace{-0.8\baselineskip}
\end{table*} 
A key aspect of our problem setting is the trade-off between extending document lengths and increasing the number of documents on the SERP.
Importantly, this trade-off is present even when the attractiveness of a document is independent of its length.
As illustrated in Table~\ref{tab:example}, in a simple three-slot ranking with documents $A$, $B$ and $C$, both rankings $[(A, 2), (B, 1)]$ and $[(B, 1), (A, 2)]$ are optimal in terms of \ac{EA} under different discount weights $\theta$.
Notably, in this example, under $\theta_2$ extending the second item into the third slot makes the user slightly more likely to observe it compared to the first item, and thus, placing the more relevant item second leads to a higher value of the \ac{EA} metric.
However, the observation that under $\theta_2$ the most attractive document $A$ should be placed second, contradicts the \ac{PRP}~\citep{robertson1977Probability}, which states that placements should be in decreasing order of relevance.
We thus conclude that \ac{PRP} does not hold in the variable document presentation length setting.

Moreover, first choosing document lengths to maximize the exposure can similarly lead to suboptimal results.  
Whilst in the above example the optimal length for $A$ is $2$ under both values of $\theta$, had document $B$ not been available (i.e., only options are $A$ and $C$), the optimal strategy in both cases would have instead been to present document $A$ at length $3$.
Therefore, the above example proves that the task of choosing document length cannot be separated into separate ranking and document-specific length selection tasks:
\begin{theorem}
Let $\pi^*(l, r\mid d, D)$ be the optimal policy for jointly choosing $r$ the rank of document $d$ and $l$ its presentation length.
The optimal policy cannot be separated into a separate ranking policy and length selection policy:
{\fontsize{8.3pt}{10.14pt}
\begin{equation}
\forall D, \not\exists (\pi_{r},   \pi_{l}), \forall (d,l,r): \;  \pi^*(l, r\mid d, D) = \pi_{r}(r\mid d, D) \pi_{l}(l \mid d).
\end{equation}}
\end{theorem}
\begin{proof}
Assume documents and position weights as in Table~\ref{tab:example}.
The optimal length policy for document $A$ is then $\pi_{l}(l=2 \mid A)=1$.
If we only change the relevance of $B$ to $0$, the optimal length policy for $A$ is $\pi_{l}(l=3 \mid A)=1$, thus there is no universal optimal $\pi_{l}$.
\end{proof}

At first glance, it may seem feasible to solve the task with dynamic programming, but a further difficulty is that the optimal ranking for $K$ slots is not necessarily an optimal sub-ranking for $K' > K$ slots.
This can be also seen in Table \ref{tab:example}: $B$ should be placed 
ahead of $A$ under $\theta_2$ and $K=3$, but the order is reversed under $K=2$ (i.e., $[(A,1), (B,1)]$).
Nevertheless, optimal rankings should vary document lengths due to the changes in exposure (and attractiveness) arising in this setting.
For instance, we can see that under $K=3$ the optimal single-length solutions $[(A, 1), (B, 1), (C,1)]$, $[(A, 2)]$ and $[(A, 3)]$ are worse than the best mixed document length solution.
Finally, we note that, whilst not shown in the above example, the optimal orderings may again be different for other values of $\rho$ that maintain the relative order $A\succ B\succ C$ but not the ratio of their $\rel$ values, e.g., if $\rho(B)=0.01$ then $(A,3)$ would always be preferred.

Altogether, whilst it appears crucial to present documents at different lengths and choose their ordering accordingly, the above examples reveal that there appears no straightforward optimal solution -- moreover, we proved that this task cannot be broken down into separate ranking and length-selection policies.
Accordingly, the remainder of this work proposes a framework for jointly optimizing document ordering and lengths.

\section{Method: Variable Length Plackett-Luce}
\label{sec:methodology}

\subsection{A distribution over rankings and lengths}
\label{subsec:vpl}

We propose to build a novel model on top of the existing \ac{PL} model for ranking (Equation~\ref{eq:PL}); whereas the original only models a distribution over rankings, our novel model captures a distribution over both rankings and presentation lengths.

Our novel model also follows sequential sampling of document placements, but instead of sampling single documents $d$, document-length pairs $(d,l)$ are sampled, where $l$ is the presentation length of $d$.
To avoid a document from appearing multiple times in the same ranking, the sampling of a $(d,l)$ pair is followed by the removal of all other pairs that involve $d$.
For our notation, we use the binary indicator variable $\mathds{1}[d_i \not\in y_{1:i-1}] = \mathds{1}[\not\exists l : (d_i, l) \in y_{1:i-1}]$ to indicate earlier placements of a document.
Another difference with standard ranking is that the lengths of placed documents affect how many can be displayed in the $K$ slots: specifically, a pair $(d,l)$ should not be sampled if fewer than $l$ slots are unfilled at the end of a partial ranking.
Thus, we also introduce the indicator $\tau(s_i + l_i) = \mathds{1}[s_i + l_i \leq K + 1] = \mathds{1}[\text{len}(y_{1:i-1}) +l_i \leq K]$ to indicate if there are enough slots left for the placement of a document-length pair.

This notation then lets us define our novel \emph{\ac{VLPL} distribution}:
\begin{align}
  \label{eq:vpl}
  \pi(y) &= \prod_{i=1}^{|y|}\pi(y_i \mid y_{1:i-1}) = \prod_{i=1}^{|y|}\pi(d_i, l_i \mid y_{1:i-1}),\\
  \label{eq:vpl2}
  \pi(y_i \mid y_{1:i-1})  &= \frac{
    e^{m(d_i, l_i)} \one{\substack{d_i \not\in y_{1:i-1}}} \tau(s_i + l_i)
  }{
    \sum_{d'=1}^{|D|}\sum_{l'=1}^{L}
    e^{m(d', l')} \one{\substack{d' \not\in y_{1:i-1}}} \tau(s_i + l')}.
\end{align}
Analogous to the standard \ac{PL} ranking model (Equation~\ref{eq:PL}), the placement probability $\pi(d, l \mid y_{1:i-1})$ only depends on the score $m(d,l)$ and the scores of other eligible unplaced $(d'\!,l')$ pairs.
However, our distribution differs in what pairs are considered eligible, as placing one pair can make several other pairs ineligible to avoid duplicate placements and space constraints.
This goes beyond the \emph{sampling without replacement} of the standard \ac{PL} model, which can also be seen a special case $L=1$ of our novel variable length model.

\subsection{Variable length Plackett-Luce optimization}

Previous work on Plackett-Luce optimization has found efficient techniques for the standard Plackett-Luce model~\citep{oosterhuis2021computationally, oosterhuis2022LearningtoRank}.
Fortunately, this existing approach can be re-used for our novel variable length model.
In particular, we can apply the derivation steps corresponding to Equations $11$-$23$ of \citet{oosterhuis2021computationally}, with only minor changes to match our model and objective.
For brevity and due to its high similarity, we do not repeat the full derivation here, but report the following resulting gradient:
\begin{align}
  \label{eq:pi2}
  &\dm\Rq 
  = \sum_{d,l=1}^{|D|,L} \Big[\dm m(d,l)\Big] \E{y} \Bigg[
    \underbrace{\sum_{i=\rr(d,l)+1}^{|y|} \theta(s_i,l_i) \rel(d_i, l_i)}_{\text{future reward after placement}} \\
      &+ \underbrace{\sum_{i=1}^{\rr(d)} \pi(d,l\mid y_{1:i-1}) \Big[
        \theta(s_i,l) \rel(d, l) - \sum_{x=i}^{|y|} \theta(s_x,l_x) \rel(d_x, l_x)
        \Big]}_{\text{expected reward minus the risk of placement}}
  \Bigg]
\end{align}
Importantly, analogous to \citet{oosterhuis2021computationally}, the gradient of $\Rq$ w.r.t. the score of a single document-length pair can be broken down into three terms: 
\begin{enumerate*}[label=(\roman*)]
\item the future reward observed following the placement, 
\item plus the expected direct reward of the placement, 
\item minus the expected risk imposed by the pair before it is placed.
\end{enumerate*}

Unsurprisingly, compared to PL-Rank~\citep{oosterhuis2021computationally}, the gradient in Equation~\ref{eq:pi2} is calculated for each $(d,l)$ pair, with the reward $\rel$ and weights $\theta$ now depending on the observed document lengths.
However, a subtle but crucial difference is that we also differentiate between $\rr(d,l)$, which 
is the rank at which $d$ is observed at length $l$, and $\rr(d)$, which is the rank of $d$ at any length (the rank function is set to $|y|$ if its input is not observed).
The distinction between $\rr(d)$ and $\rr(d,l)$ is important as, due to the ranking size limit $K$, placing the document $d$ at distinct lengths may change which $(d'\!,l')$ can follow it.
Therefore, the future reward term for each length $l$ of $d$ is thus included only when that particular pair $(d,l)$ is observed.
The risk term and the expected reward, however, accumulate up until $d$ has been shown at any length, proportional to $\pi(d, l \mid y_{1:i-1})$, after which all lengths of $d$ become ineligible for placement.
If a presentation length $l$ is too long, these terms are $0$, as $\tau(s_i+l) = 0 \rightarrow \pi(d, l \mid y_{1:i-1})=0$.

The gradient formulation in Equation~\ref{eq:pi2} can be turned into an estimator by replacing the expectation with an average over sampled rankings.
For brevity, we omit the full equation but name this approach VLPL-1.
To our knowledge, VLPL-1 is the first method to jointly optimize rankings and lengths in the variable document presentation length setting.

The remainder of this section describes how to efficiently sample rankings from the VLPL distribution (Section~\ref{sub:sampling_rankings}).
Subsequently, in Section~\ref{sub:VLPL2},  we propose a more sample-efficient estimator  through a smarter sampling technique for re-using samples.
Finally, we discuss how VLPL can be applied in an in-processing and post-processing manner (Section~\ref{sub:VLPLapplications}).

\subsection{Efficiently sampling rankings}
\label{sub:sampling_rankings}
Whilst it may appear that variable length rankings are substantially more difficult to sample than standard rankings~\citep{oosterhuis2021computationally}, this can actually be done with a comparable computational complexity.
The trick is to first sample an invalid ranking $y'$ and then transform it with a function $f$ such that $f(y')$ is a valid variable length ranking.
We start by defining $y'$, which, in contrast with the variable length rankings as defined before, places $d$ once at every possible length $l$, whilst ignoring the maximum ranking length $K$.
Still, $y'$ does not allow the placement of duplicate documents of the same length.
This gives us the following probability for document placements:
\begin{equation}
\pi'(d, l \mid y_{1:i-1}') = \frac{
e^{m(d, l)} \one{(d,l) \not\in  y_{1:i-1}' }
}{
\sum_{d'=1}^{|D|}\sum_{l'=1}^{L} e^{m(d', l')} \one{(d',l') \not\in  y_{1:i-1}'}
}.
\end{equation}
Let $d_i'$ and $l_i'$ indicate the document at position $i$ in $y'$ and its length respectively.
The probability of the ranking $y'$ is then simply the product of each placement probability:
\begin{equation}
\pi'(y') = \prod^{|D| L}_{i=1}p(d_i', l_i' \mid y_{1:i-1}').
\end{equation}
We note that every document-length pair is placed in $y'$ and thus $|y| = DL$.
Obviously, this makes $y'$ an invalid variable length ranking as each document will appear $L$ times in the ranking, and thus it is guaranteed to exceed the maximum length of $K$.

Our next step is to use a function $f^K$ that turns $y'$ into a valid ranking of length $K$, such that we end up with the same ranking distribution: $\pi(f(y') = y) = \pi(y)$.
Our choice of $f$ is a function that sequentially takes the first document in $y'$ that does not violate either of the variable ranking constraints, until length $K$ is reached:
\begin{equation}
\begin{split}
f(y')_i = y_j' \; &\text{where }
j = \arg\min_{j'} j' \in \{1,2,\ldots,DL\}
\\ 
&\text{s.t. } d_{j'}' \not\in f(y')_{1:i-1} \land \text{len}(f(y')_{1:i-1}) + l_{j'}' \leq K.
\end{split}
\end{equation}
In other words, $f(y')_1 = y_1'$, then $f(y')_2$ is the next document that would not violate the validity of the ranking when placed, and so forth.
To show that this yields the same distribution as $\pi$, 
we first note that the probability of a document placement: $\pi(f(y')_i = d \mid f(y')_{1:i-1})$ is the probability of it being a valid option and being ranked in front of all other valid options that are still available after $f(y')_{1:i-1}$.
Following Luce's choice axiom~\citep{luce1977Choice}, the ordering probabilities of the existing options in a \ac{PL} model are not changed under the introduction of additional alternatives, i.e., $\text{Pr}(a \succ b \mid \{a,b\}) = \text{Pr}(a \succ b \mid \{a,b,c\})$.
Thus, the availability of invalid options in $y'$ does not change the relative ranking of valid options, and therefore $\pi(f(y')_i = (d,l) \mid f(y')_{1:i-1}) = \pi(y_i = (d,l) \mid y_{1:i-1})$ when $f(y')_{1:i-1} = y_{1:i-1}$.
By extension, $\pi(f(y') = y) = \pi(y)$.

\subsection{Extending VLPL for sample-efficiency}
\label{sub:VLPL2}

As previously suggested, VLPL-1 may have suboptimal sample-efficiency when calculating the future reward.
In order to address this, we first start by 
defining the total placement reward as:
\begin{align}
  g(y, i) &= \sum_{x=i}^{|y|} \theta(s_x,l_x) \rel(d_x, l_x).
\end{align}
This allows us to re-formulate Equation~\ref{eq:pi2} as:
\begin{align}
  \label{eq:gh_g}
  g(y, d, l) &= g(y, \rr(d,l)+1),\\
  \label{eq:gh_h}
  h(y, d, l) \hspace{-0.25mm}&= \hspace{-1mm}\sum_{i=1}^{\rr(d)} \hspace{-0.25mm}\pi(d,l\mid y_{1:i-1}) \big[\theta(s_i,l) \rel(d, l) \hspace{-0.25mm}-\hspace{-0.25mm} g(y, i)\big].
\end{align}
The expected future reward with our sampling scheme is then:
{\fontsize{8.5pt}{10.4pt}
\begin{equation}
\begin{split}
  \E{} \Bigg[
    \sum_{i=\rr(d,l)+1}^{|y|} \hspace{-3mm}\theta(s_i,l_i) \rel(d_i, l_i)
  \Bigg] \hspace{-0.25mm}=\hspace{-0.25mm}  \E{} \big[
    g(y, d, l)
  \big] \hspace{-0.25mm}=\hspace{-0.25mm}  \E{} \big[
    g(f(y'),d, l)
  \big]
  .
\end{split}
\end{equation}}
To increase the sample-efficiency, we want to reuse each $y'$.
For this purpose, we introduce the following function:
$f(y'\shortleftarrowtwo d,l)$, which is the same transformation as $f$, except when $d$ is being selected its length is set to $l$ (instead of the actually sampled length).
Furthermore, we introduce the variable $l(d, y')$, denoting the length of the first instance of $d$ that is included in $f(y')$.
Our aim is then to use samples arising from the sampling distribution of $l(d,y')$ to estimate the expected value under $l(d, y') = l$.
We therefore introduce the following importance weight:
{\fontsize{8.5pt}{10.4pt}
\begin{align}
\label{eq:importance_weight}
p(d, l, y') \hspace{-0.01cm}&=\hspace{-0.01cm} \dfrac{p\big(l(d, y') = l\big)}{\sum_{l'=1}^L p\big(l(d, y')=l'\big)}\hspace{-0.01cm}=\hspace{-0.01cm}\frac{
e^{m(d,l)}\tau'(d, l, y')}{
\sum_{l'=1}^L e^{m(d,l')}\tau'(d, l', y')
},
\end{align}}
where $\tau'(d, l, y') = \mathds{1}[{\text{len}(f(y')_{1:\rr(d\mid f(y'))-1}) + l \leq K}]$ is true only for the lengths at which $d$ can be placed when it is observed in $f(y')$.
This then enables the following reformulation:
\begin{equation}
\E{y'} \big[
    g(f(y'), d, l)
  \big]
=
\E{y'} \big[p(d, l, y')
    g(f(y' \shortleftarrowtwo d,l), d, l)
  \big],
\end{equation}
That is, we can use rankings $y'$ with any observed length $l(d,y')$, process them for the target length $l$ with $f(y\shortleftarrowtwo d, l)$ and reweigh the resulting future reward with $p(d, l, y')$.
Repeating this for all valid lengths at each position yields the gradient estimate:
\begin{equation}
\begin{split}
  \label{eq:pi3_fgh}
  \dm\Rq 
  &= \sum_{d,l=1}^{|D|,k} \Big[\dm m(d,l)\Big] \\
  &\;\, \cdot\E{y'} \big[p(d, l, y')
    g(f(y' \shortleftarrowtwo d,l), d, l) + h(f(y'), d, l)
  \big].
\end{split}
\end{equation}

Furthermore, whilst calculating a separate $f(y' \shortleftarrowtwo d,l)$ for each $(d,l)$ pair may seem time-consuming, in practice it is actually not needed.
We first define $\Delta \in [1-L, 2-L,\ldots,L-1]$ as the difference between the target length $l$ and observed length $l(d,y')$.
Importantly, for a source ranking $y'$ which after applying $f(y')$ includes $(d_i, l_i)$ and $(d_j, l_j)$, the adjusted rankings $f(y'\shortleftarrowtwo d_i, l_i+\Delta)$ and $f(y'\shortleftarrowtwo d_j, l_j+\Delta)$ in fact share the same tail, i.e., $z=\max(i,j)\rightarrow f(y'\shortleftarrowtwo d_i, l_i+\Delta)_{z+1:} = f(y'\shortleftarrowtwo d_j, l_j+\Delta)_{z+1:}$.
We can thus use the same ranking to calculate the future reward for multiple such shifts, respectively reweighting each with $p(d_1, l_1+\Delta, y'), p(d_2, l_2+\Delta, y'), \ldots, p(d_{|f(y')|}, l_{|f(y')|}+\Delta, y')$.
Moreover, these future placements would also be the same if instead we still placed $d_i$ at the original observed length $l_i$ but then reduced the total number of available slots in the ranking by $\Delta$, i.e. $f(y'\shortleftarrowtwo d_i, l_i+\Delta)_{z+1:} = f^{K-\Delta}(y', d_i, l_i)_{z+1:}=f^{K-\Delta}(y')_{z+1:}$, where $f^{K'}(y')$ produces a ranking $y$ of length $K'$ instead of $K$.

As $\Delta$ can take on values between $1-L$ and $L-1$ (corresponding to replacing the longest length with the shortest and vice-versa), the future reward always arises from one of $2L-1$ rankings $f^{K+\Delta}(y')$.
However, whilst $f^K(y' \shortleftarrowtwo d, l+\Delta)$ and $f^{K-\Delta}(y' \shortleftarrowtwo d, l)$ lead to the same placements after $d$, the slot positions $s_i$ in the latter would be shifted down by $\Delta$.
To maintain the consistency in the reward between the two cases and to ensure the correct range of $\theta$'s inputs,
we define the adjusted total reward as:
{\fontsize{8.5pt}{10.4pt}
\begin{equation}
  \label{eq:g_prime}
  g'(y, i, \Delta) = \sum_{x=i}^{|y|}\mathds{1}[s_x+\Delta\geq 1]\theta(s_x+\Delta, l_x)\rel(d_x, l_x).
\end{equation}}

We note that the indicator variable will always be one when calculating the future reward $g'(y, i+1, \Delta)$ for $1\leq l(d_i, y')+\Delta\leq L$.

Finally, as shown by \citet{oosterhuis2022LearningtoRank}, the computation of a gradient estimator for the \ac{PL} distribution can be further accelerated by utilizing helper variables, which are first pre-computed and re-used for all positions.
Along with applying Equation~\ref{eq:g_prime}, we can thus express the terms in Equation~\ref{eq:pi3_fgh} as:
{\fontsize{8.5pt}{10.4pt}
\begin{align}
  \text{PR}^{\Delta}_{y,i} &= g'(f^{K-\Delta}(y), i, \Delta), \quad 
  \text{PR}^{\Delta}_{y,(d)} = \text{PR}^{\Delta}_{y,\rr(d)+1},\\
  \text{RI}_{y,i} &= \sum_{x=1}^{\min(|y|, i)}\frac{\text{PR}^0_{y,i}}{\sum_{d'=1}^{|D|} \sum_{l'=1}^L e^{m(d',l')}\mathds{1}[\substack{d\not\in y_{1:x-1}}]\tau{[s_x+l']}},\\
  \text{DR}_{y,i,l} &= \sum_{x=1}^{\min(|y|, i)}\frac{\theta(s_x, l)}{\sum_{d'=1}^{|D|}\sum_{l'=1}^L e^{m(d',l')}\mathds{1}[\substack{d\not\in y_{1:x-1}}]\tau{[s_x+l']}},\\
  \text{RI}_{y,(d,l)} &= \text{RI}_{y, \min(\rr(d), \rr(l))}, \quad
  \text{DR}_{y,(d,l)} = \text{DR}_{y, \min(\rr(d), \rr(l)), l},
\end{align}}
where $\rr(l)$ is $\sup i\hspace{-0.2mm}:\hspace{-0.2mm} \text{len}(f(y')_{1:i-1})+l\leq K$, i.e., the last position at which a document of length $l$ can be placed, and $\rr(d)$ is the rank of $d$ in $y^0=y$. This, in turn, allows us to reformulate Equation~\ref{eq:pi2} as:
{\fontsize{8.3pt}{9pt}
\begin{equation}
\begin{split}
  \label{eq:pi3_plrank3}
  \dm\Rq = &\sum_{d,l=1}^{|D|,L} \Big[\dm m(d,l)\Big]\E{y'}\Big[
    p(d, l, y')\text{PR}_{y', (d)}^{l-l(d,y')}\\ 
    &\qquad + e^{m(d,l)}\Big(
      \rel(d,l)\text{DR}_{f(y'), (d,l)} - \text{RI}_{f(y'),(d,l)}
    \Big)
  \Big].\!\!\!
\end{split}
\end{equation}}
This gradient can be approximated using the sampling procedure for $\pi'$ described in Subsection~\ref{sub:sampling_rankings}, yielding our estimator VLPL-2:
{\fontsize{8.3pt}{9pt}
\begin{equation}
\begin{split}
  \label{eq:pi3_plrank3}
 \dm\Rq = &\sum_{d,l=1}^{|D|,L} \dm m(d,l)\sum_{i=1}^N \dfrac{1}{N}\Big[
    p(d, l, y'^{(i)})\text{PR}_{y'^{(i)}, (d)}^{l-l(d,y'^{(i)})}\\ 
   & + e^{m(d,l)}\Big(
      \rel(d,l)\text{DR}_{f(y'^{(i)}), (d,l)} - \text{RI}_{f(y'^{(i)}),(d,l)}
    \Big)
  \Big].\!\!\!
\end{split}
\end{equation}}
This estimator has higher sample-efficiency than VLPL-1 as it shares the future reward across different lengths of the same document. 
The risk and the expected direct reward terms RI and DR are still calculated using the standard ranking $y=y^0$ of length $K$.
The exact steps of our approach are also outlined in Algorithm~\ref{alg:one}.

Finally, we note that VLPL-1 can similarly be implemented with PL-Rank-3-like sample-efficiency by using the reward {\fontsize{9pt}{1pt}PR$^0_{y', \rr(d,l)+1}$} in place of {\fontsize{9pt}{1pt}PR$^\Delta_{y', \rr(d)+1}$} and setting $p(d, l, y')$ to $1$.

\begin{algorithm}[hbt!]
\renewcommand{\algorithmicrequire}{\textbf{Input:}}
\caption{VLPL-2 Gradient Estimator}\label{alg:one}
\begin{algorithmic}[1]
\STATE \textbf{Input:} Items: $D$; Scores: $m$; Relevances: $\rel$; Position weights: $\theta$; Slot number: $K$; Item length limit: $L$; Sample number: $N$.
\STATE $\{y'_{(1)}, y'_{(2)}, \ldots, y'_{(N)}\} \gets \text{SampleRankings(m, N, K, L)}$
\STATE $\text{Grad} \gets \mathbf{0} \in {\mathbb{R}^{|D| \times L}}$ 
\STATE $S_1, \ldots, S_L \gets \sum_{d=1}^{|D|} \exp(m(d,1)), \ldots \sum_{d=1}^{|D|}\exp(m(d,L))$
\FOR{$j \in [1,2,\ldots,N]$}
  \STATE \textit{\small // Each $y^x=[(d^x_1, l^x_1),\ldots, (d^x_{|y^x|},  l^x_{|y^x|})], s_i^x=1+\text{len}(y_{1:i-1}^x)$}
  \STATE $y^{-L+1}, \ldots, y^{L-1} \gets f^{\substack{K+L-1}}(y'_{(j)}), \ldots, f^{\substack{K-L+1}}(y'_{(j)})$
  \STATE{$S_1',\ldots,S_L' \gets S_1, \ldots, S_L$}
  \STATE{$\text{PR}^{-L+1}_{K+L}, \ldots, \text{PR}^{L-1}_{K+L} \gets 0, \ldots, 0$}
  \FOR{$i \in [K+L-1, K+L-2, \ldots, 1]$}
    \FOR{$\Delta \in [-L+1, \ldots, L-1]$}
      \STATE{$u = \mathds{1}[|y^\Delta| \geq i]\mathds{1}[s_i +\Delta\geq 1] \rel(d^\Delta_{i}, {l^\Delta_i}) \theta(s^\Delta_i + \Delta, l^{\Delta}_{i})$}
      \STATE{$\text{PR}^\Delta_{i} \gets \text{PR}^\Delta_{i+1} + u$} 
    \ENDFOR
  \ENDFOR
  \STATE{$\text{RI}_0, \text{DR}^1_0, \ldots, \text{DR}^L_0 \gets (0, 0, \ldots, 0)$}
  \FOR{$i \in [1, \ldots, K]$}
    \STATE{$S' = \sum_{l=1}^L{S'_l}$}
    \STATE{$\text{RI}_i \gets \text{RI}_{i-1} + \text{PR}_i^0 / S'$}
    \FOR{$l \in [1, \ldots, L]$}
      \STATE{\text{DR}$^l_i \gets \text{DR}^l_{i-1} + \theta(s^0_i, l) / S'$}
      \STATE{$S'_l\gets \mathds{1}[l \leq K - (s^0_{i} + l^0_{i}) + 1] (S'_l - \exp(m(d^0_i, l)))$}
      \STATE{$p_i^l \gets \exp{(m(d^0_i,l))} \mathds{1}[\substack{l \leq K - s_i^0 + 1}]$}
    \ENDFOR
  \ENDFOR
  \FOR{$d \in [1, \ldots, |D|]$}
    \FOR{$l \in [1, \ldots, L]$}
      \STATE{$r_1 \gets \min(\text{rank}(d, y^0),  K)$}
      \STATE{$r_2 \gets \min(r_1, \sup i: \{l \leq K - s^0_i + 1\})$}
      \STATE{$u \gets  \exp(m(d,l))[\rel(d,l)\cdot \text{DR}_{r_2}^l \hspace{-0.4mm}-\hspace{-0.4mm} \text{RI}_{r_2}]$}
      \STATE{$\text{Grad}(d, l) \gets \text{Grad}(d, l) + p_{r_1}^l / \sum_{l'=1}^{L} p_{r_1}^{l'} \cdot \text{PR}_{r_1+1}^{l-l^0_{r_1}} + u $}
    \ENDFOR
  \ENDFOR
\ENDFOR
\STATE \text{\textbf{return}} $\text{Grad} / N$
\end{algorithmic}
\end{algorithm}

\subsection{Using VLPL as in- or post-processing}
\label{sub:VLPLapplications}

Existing \ac{LTR} work mainly aims to optimize machine learning models for the ranking task~\citep{oosterhuis2021computationally,oosterhuis2020topkrankings,buchholz2022Lowvariancea};
we will refer to this as an \emph{in-processing} application of \ac{LTR}.
However, \ac{LTR} gradients can also be used to optimize a ranking directly, i.e., when $m$ is a look-up table instead of a model.
This also allows for a \emph{post-processing} application of \ac{LTR}, and by extension VLPL, where the reward values $\rel$ are predicted, and subsequently, the ranking is optimized for these predictions.
Post-processing  is not very popular for standard ranking, as the optimal ranking simply sorts according to $\rel$, thus there is little added value, while it adds substantial costs during inference.
However, in our variable document length setting the \ac{PRP} does not hold and, therefore, post-processing could actually result in better rankings than is possible with a standard scoring model.
For similar reasons, post-processing ranking has also been used for ranking fairness \citep{caton2024Fairness}.
In this work, we evaluate VLPL both applied as an in-processing method and a post-processing method.

\section{Experimental Setup}
\label{sec:experimental_setup}

Our aim is to better understand the variable presentation length setting and evaluate our proposed methods for this task.
To begin, as discussed in Section~\ref{sub:PRP}, there is no known method that can provide the optimal ranking and lengths in this setting,
even when all the attractiveness values $\rel$ for the query are available (i.e., the oracle setting).
We are therefore interested in determining whether VLPL can find high reward rankings, both when the $\rel$ values are given or unknown.
Thus our first two research questions are:
\begin{enumerate}[align=left, label={\bf RQ\arabic*},leftmargin=*]
\item When provided with true attractiveness labels, does VLPL learn rankings with a higher \ac{EA} metric than the baselines?
    \label{rq:oracle}
\item Does VLPL also learn higher \ac{EA} rankings when it does not have access to the true attractiveness?
    \label{rq:performance}
\end{enumerate}
Furthermore, as it is unclear what optimized variable length rankings look like, our third research question asks:
\begin{enumerate}[align=left, label={\bf RQ\arabic{*}},leftmargin=*, resume]
\item What document lengths does VLPL place in each position for different choices of $\theta$?
    \label{rq:rankings}
\end{enumerate}
Finally, since VLPL uses samples for its gradient estimation, it seems important to understand its sample-efficiency.
Our final research question is therefore:
\begin{enumerate}[align=left, label={\bf RQ\arabic*},leftmargin=*, resume]
    \item Does VLPL have comparable performance with a lower number of sampled rankings?
    \label{rq:N}
\end{enumerate}
Thus, our research questions cover different aspects of VLPL's performance and the variable presentation length setting.

\subsection{Datasets and length-specific attractiveness}

To our knowledge, there are no datasets featuring different presentation lengths for the same document that can be used to answer our research questions.
We thus use \textit{Yahoo! Webscope} \citep{Chapelle2011} and \textit{MSLR-WEB30k} \citep{qin2013introducing} datasets, consisting of queries and associated documents.
Each query-document pair is represented by a feature vector and an expert relevance judgement label $R\in\{0,1,\ldots,4\}$.
Yahoo contains 29{,921} queries with an average of 24 documents per query, whilst MSLR consists of 30{,}000 queries with an average of 125 documents per query.
For computational efficiency we focus on the queries with $250$ or fewer documents, retaining all queries in Yahoo and over $95\%$ of queries in MSLR.

To generate $L$ attractiveness labels for every document, we first remap the labels in Yahoo to $5L$ non-overlapping equally sized bins. For a document $d$ with original label $R$, the attractiveness $\rel$ is higher for longer document presentation lengths $l$, yet is still primarily determined by $R$: $\rel \in [(RL+l-1)/(5L), (RL+l)/(5L)]$.
On the other hand, due to the higher number of documents per query in MSLR, we emphasize the most relevant documents and double each bin's size relative to the last, renormalizing their sum.

In order to determine $\rel(d,l)$, we then follow~\citet{fang2019intervention}, generating $L+1$ scores using query-document features.
$L$ of these are normalized with other scores in $(d,l)$'s bin and map to a quantile within that bin's range, yielding $\rel(d,l)$.
Furthermore, as some documents may not actually benefit from a longer presentation~\citep{sweeney2006Effective}, we randomly reorder the $\rel$ values of different lengths for half of the documents determined by the $L+1$'th score.

\subsection{Models, baselines and training}

For our in-processing models, we use VLPL-1 and VLPL-2 to calculate the gradient of $\dm \Rq$ in order to train a \ac{MLP} that learns  $L$ scores $m(d,l)$ in Equation~\ref{eq:vpl2} for each query-document pair based on its features.
In contrast, for the post-processing approach, we use the binary cross-entropy loss to first train an \ac{MLP} to predict $\rel$ (instead of $m$).
Then for each test query, we instantiate a look-up table with a single score for each $(d,l)$ pair, and use 
VLPL-1/2 to learn these scores to maximize Equation~\ref{eq:opt_objective} specifically for that query, 
using the outputs of \ac{MLP} in place of  $\rel$.

We contrast our four VLPL models with the following baselines.
The first method family, referred to as sort-$l$, similar to the VLPL post-processing, use an \ac{MLP} to estimate $\rel$ for a single length $l$ to sort and place documents of only that length.
In addition, we train $L$ separate single-length PL-Rank-3 models using in-processing, which we denote as PLR-3-$l$.
For length-adaptive baselines, the greedy model uses MLP estimates of $\rel$ to repeatedly select an eligible $(d,l)$ pair maximizing $\theta(s,l)\rel(d,l)$.
However, the above policy may also prefer one longer document over multiple shorter yet only slightly less relevant ones.
The slot-avg model thus instead ranks $(d,l)$ pairs on $\theta(s,l)\rel(d,l) / l$, in effect choosing the document with the highest per-slot expected reward.
Finally, we also evaluate the above models with full access to the true relevances (oracle setting). 
In this case the fixed-length sorting models are an upper bound to PL-Rank-3, which is therefore omitted.

All models contain $1$-$5$ hidden layers of size $\{16, 32, \ldots, 128\}$ and sigmoid activations, trained with Adam optimizer \citep{kingma2017Adam} with learning rate, $L2$ regularization in $\{0, 1^{-6}, 1^{-5}, \ldots, 1\}$ and dropout in $[0, 0.8]$.
We report test set results over $5$ independent runs performed under identical circumstances for the best configuration of each model, selected on the validation set results of the first run.
VLPL trained with $N=10{,}000$ sampled rankings per update step.
All models trained in PyTorch \citep{paszke2019pytorch} on a single NVIDIA RTX A5000 GPU.

\subsection{Evaluation and position bias}

We evaluate all models using the \ac{EA} objective in Equation~\ref{eq:opt_objective}. 
We fix the number of slots $K$ at $30$ and maximum document length $L$ at $3$, representing one possible production setting.
Under the standard rank-based position bias model \citep{oosterhuis2020topkrankings}, users are assumed to observe the contents of each position with a probability $\theta(i)$ that only depends on the position.
As in our setting a document can span multiple positions, we consider the document to be observed as long as any of its slots are observed, yielding $\theta(s,l) = 1 - \prod_{i=s}^{s+l-1} \left(1 - \theta(i, 1)\right)$, where $
\theta(i, 1)$ is equivalent to $\theta(i)$ in Equation \ref{eq:standard_objective} -- the standard observation probability for rank $i$ \citep{joachims2003evaluating}.
We evaluate two choices of $\theta(i)$ -- the slowly decaying DCG weights $\theta_{\text{DCG}}(i) = 1/\log_2(i+1)$ and the steeper inverse rank $\theta_{\text{rank}^{-1}}(i) = 1/i$, corresponding to DCG and the sum of reciprocal ranks if $\forall l=1$.

\section{Results}
\label{sec:results}
\subsection{Ranking with known relevance}
\label{sub:results_oracle_setting}

We start by addressing the first research question (\ref{rq:oracle}): whether VLPL is able to more efficiently order documents and their lengths to achieve higher reward when compared to the baselines, when the attractiveness of each document is known.
Table~\ref{table:results_oracle} shows the performance in terms of \ac{EA} in this oracle setting where the number of slots $K=30$ and maximum document length is $L=3$ slots.

We can see that adaptively selecting document lengths has a very large impact on the \ac{EA} metric.
Both variable length baselines achieve a substantial and consistent improvement over the single-length baselines, with the greedy model performing better under $\theta_{\text{DCG}}$ and the slot-avg model with $\theta_{\text{rank}^{-1}}$.
Different strategies also achieve the best performance  among the single-length models in different settings, highlighting the importance of tailoring the document presentation to the setting where the ranking is shown.

On the other hand, both VLPL-1 and VLPL-2 demonstrate a very large and statistically significant improvement over \emph{all} models across every setting, with a relative performance increase in \ac{EA}$_\text{DCG}$ over the best baseline  on Yahoo exceeding $10\%$.
As such, VLPL is clearly able to generate highly attractive rankings in this setting.
Section~\ref{sub:sample_efficiency} further compares the performance of VLPL models.

Overall, we observe that adaptively choosing document presentation length is very important, with particularly high performance shown by VLPL.
We thus answer \ref{rq:oracle} positively: when true attractiveness is known, compared to the baselines, VLPL is able to find significantly more attractive variable presentation length rankings.

\begin{table}[tb]
 \centering
 \caption{
\ac{EA} for $K=30$ and $L=3$ in the oracle setting with DCG and inverse rank slot weights. Results are averages over five independent runs; standard deviations in parentheses. Best result in each task denoted in bold. Significant improvement of a model over all (other) baselines denoted by $\triangle$ (separate one-sided Wilcoxon signed rank tests, $p\!<\!0.05$). 
\vspace{-3.4mm}
 }
 \label{table:results_oracle}
   \resizebox{\columnwidth}{!}{\addtolength{\tabcolsep}{-0.3em}\begin{tabular}{cccccc}
    \toprule
    & \multicolumn{2}{c}{\textbf{MSLR}}&\multicolumn{2}{c}{\textbf{Yahoo!}}\\
    \midrule
    \textbf{Model}& \textbf{\ac{EA}$_\text{DCG}$} & \textbf{\ac{EA}$_{\text{rank}^{-1}}$} & \textbf{\ac{EA}$_\text{DCG}$} & \textbf{\ac{EA}$_{\text{rank}^{-1}}$}\\
    \midrule
sort-1 &${1.500}\pt$(0.002)&${0.908}\pt$(0.002)&${2.832}\pt$(0.000)&${1.695}\pt$(0.000)\\
sort-2 & ${1.628}\pt$(0.005)&${0.925}\pt$(0.004)&${3.113}\pt$(0.000)&${1.680}\pt$(0.000)\\
sort-3 &${1.733}\pt$(0.010)&${0.981}\pt$(0.006)&${3.137}\pt$(0.000)&${1.662}\pt$(0.000)\\\midrule
greedy &${1.917}\pt$(0.008)&${1.086}\pt$(0.005)&${3.326}^\triangle$(0.000)&${1.778}\pt$(0.000)\\
slot-avg & ${1.911}\pt$(0.005)&${1.113}^\triangle$(0.003)&${3.282}\pt$(0.000)&${1.926}^\triangle$(0.000)\\\midrule
VLPL-1 &$\textbf{2.064}^\triangle$(0.008)&$\textbf{1.187}^\triangle$(0.005)&${3.665}^\triangle$(0.000)&$\textbf{2.031}^\triangle$(0.000)\\
VLPL-2 &$\textbf{2.064}^\triangle$(0.008)&$\textbf{1.187}^\triangle$(0.005)&$\textbf{3.670}^\triangle$(0.000)&${2.030}^\triangle$(0.000)\\
    \bottomrule
   \end{tabular}}
\end{table} 

\subsection{Ranking with learned relevance}
\label{sub:ranking_with_learned_relevance}

We have established that VLPL can effectively optimize variable presentation length rankings when documents' attractivenesses are given.
Nevertheless, for previously unseen queries these are unknown.
Accordingly, we aim to answer our second research question (\ref{rq:performance}): can VLPL be used to generalize to unseen queries?

As shown in Table~\ref{table:results_main}, the performance of all models is substantially lower compared to the oracle setting (cf.\ Table~\ref{table:results_oracle}). Here, sort$-3$ now also consistently performs the best among sort$-l$ models.
Together, these observations suggest that relevance misestimation can have a significant impact on learned rankings in this setting.

Nevertheless, both in-processing and post-processing versions of VLPL still achieve statistically significant improvements over all baselines across all settings.
We observe particularly high performance for the post-processing approaches, and especially VLPL-2.
The high performance of VLPL strongly suggests that being able to tailor document length choices to the specific document pool of each query may be beneficial in the variable length ranking setting.

Interestingly, VLPL-1 appears to generally perform better than VLPL-2 for in-processing.
As the performance of both VLPL models, and particularly VLPL-2, is higher for post-processing,
we speculate that in-processing may introduce additional noise in the future reward and which may be amplified by reward-sharing of VLPL-2.

Nevertheless, we observe that all VLPL methods perform the best across all our settings.
Therefore, we can answer~\ref{rq:performance} affirmatively: VLPL methods also achieve the highest \ac{EA} on unseen queries with learned attractiveness.

\begin{table}[t]
 \centering
 \caption{
\ac{EA} for $K=30$ and $L=3$ with learned relevance and DCG and inverse rank slot weights. Results are averages over five independent runs, with standard deviations in parentheses. Best result for each task denoted in bold. Best in-processing model in each task is underscored. Significant improvement over all baselines denoted by $\triangle$ (separate one-sided Wilcoxon signed rank tests, $p<0.05$).\vspace{-3.3mm}
 }
 \label{table:results_main}
   \resizebox{\columnwidth}{!}{\addtolength{\tabcolsep}{-0.3em}\begin{tabular}{ccccc}
   \toprule
    & \multicolumn{2}{c}{\textbf{MSLR}}&\multicolumn{2}{c}{\textbf{Yahoo!}}\\
    \toprule
    \textbf{Model} & \textbf{\ac{EA}$_\text{DCG}$} & \textbf{\ac{EA}$_{\text{rank}^{-1}}$} & \textbf{\ac{EA}$_\text{DCG}$} & \textbf{\ac{EA}$_{\text{rank}^{-1}}$}\vspace{0.5mm}\\
    \toprule
    \multicolumn{5}{c}{In-processing}\\
    \toprule
PLR-3-1 &${0.777}\pt$(0.008)&${0.421}\pt$(0.003)&${2.466}\pt$(0.001)&${1.385}\pt$(0.000)\\
PLR-3-2 & ${0.839}\pt$(0.006)&${0.458}\pt$(0.003)&${2.699}\pt$(0.001)&${1.394}\pt$(0.002)\\
PLR-3-3 & ${0.931}\pt$(0.016)&${0.522}\pt$(0.008)&${2.761}\pt$(0.001)&${1.423}\pt$(0.001)\\\midrule
VLPL-1 &${0.953}^\triangle$(0.019)&$\underline{0.533}^\triangle$(0.008)&$\underline{2.982}^\triangle$(0.021)&$\underline{1.603}^\triangle$(0.001)\\
VLPL-2 &$\underline{0.959}^\triangle$(0.018)&${0.527}^\triangle$(0.005)&${2.976}^\triangle$(0.002)&${1.590}^\triangle$(0.003)\\\toprule\multicolumn{5}{c}{Post-processing}\\\toprule
sort-1 &${0.781}\pt$(0.005)&${0.419}\pt$(0.002)&${2.469}\pt$(0.001)&${1.385}\pt$(0.001)\\
sort-2 & ${0.843}\pt$(0.006)&${0.452}\pt$(0.003)&${2.698}\pt$(0.001)&${1.404}\pt$(0.001)\\
sort-3 &${0.929}\pt$(0.015)&${0.513}\pt$(0.010)&${2.757}\pt$(0.001)&${1.424}\pt$(0.001)\\\midrule
greedy &${0.935}\pt$(0.018)&${0.492}\pt$(0.011)&${2.747}\pt$(0.001)&${1.428}\pt$(0.005)\\
slot-avg & ${0.896}\pt$(0.020)&${0.516}\pt$(0.007)&${2.600}\pt$(0.010)&${1.520}\pt$(0.001)\\\midrule
VLPL-1 &$\textbf{0.971}^\triangle$(0.020)&${0.535}^\triangle$(0.009)&${3.042}^\triangle$(0.001)&${1.627}^\triangle$(0.000)\\
VLPL-2 &$\textbf{0.971}^\triangle$(0.020)&$\textbf{0.536}^\triangle$(0.008)&$\textbf{3.049}^\triangle$(0.001)&$\textbf{1.628}^\triangle$(0.001)\\
    \bottomrule\\
   \end{tabular}}
\end{table}

\subsection{Learned rankings}
\label{sub:learned_rankings}
\begin{figure}[tb]
\centering
\rotatebox{270}{
\addtolength{\tabcolsep}{-0.3cm}
\begin{tabular}{c c c c}
\rotatebox[origin=lt]{90}{\hspace{0.4cm}\small $\theta_{\text{rank}^{-1}}$}\hspace{0.3cm} &
\includegraphics[scale=0.390]{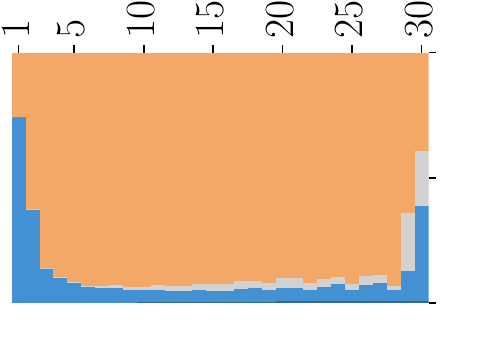} &
\includegraphics[scale=0.390]{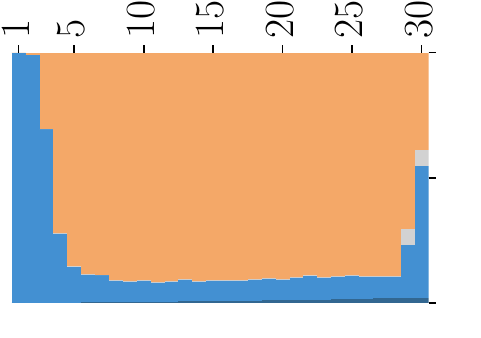} &
\includegraphics[scale=0.390]{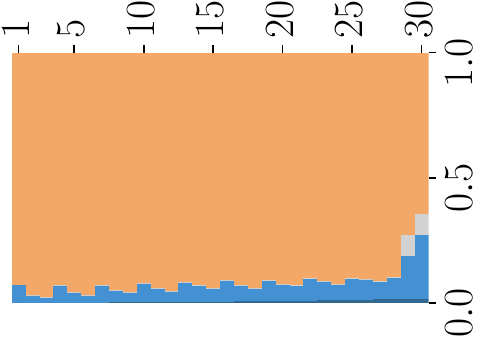}
\vspace{-0.32cm}
\\
\rotatebox[origin=lt]{90}{\hspace{-0.0cm}\small \hspace{0.1cm}oracle $\theta_{\text{rank}^{-1}}$}\hspace{0.3cm} &
\includegraphics[scale=0.390]{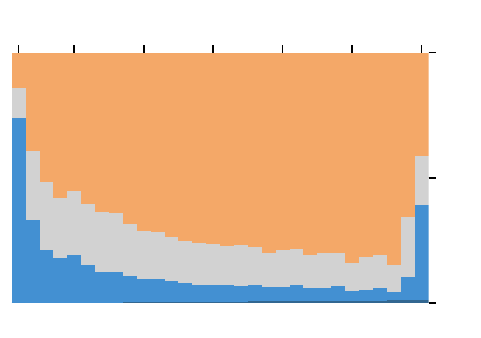} &
\includegraphics[scale=0.390]{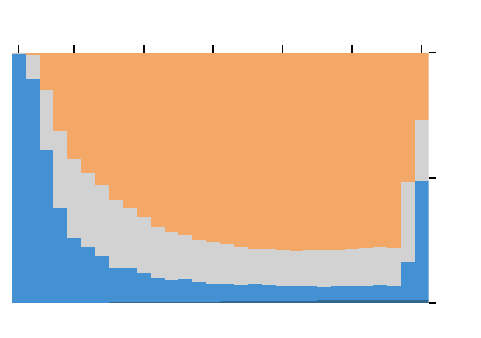} &
\includegraphics[scale=0.390]{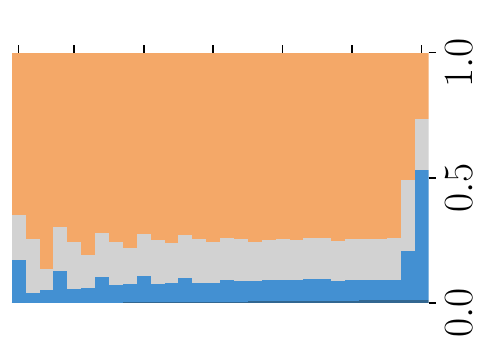}
\vspace{-0.32cm}\\
\rotatebox[origin=lt]{90}{\hspace{0.5cm}\small $\theta_\text{DCG}$}\hspace{0.3cm} &
\includegraphics[scale=0.390]{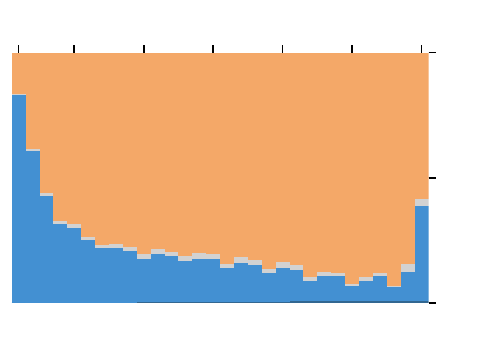} &
\includegraphics[scale=0.390]{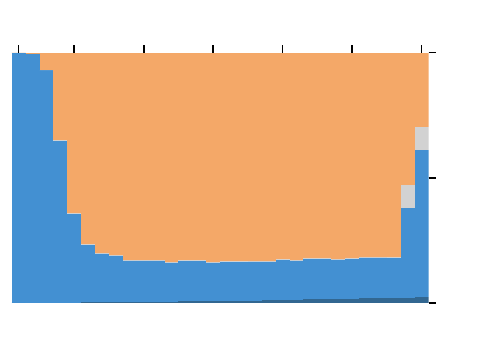} &
\includegraphics[scale=0.390]{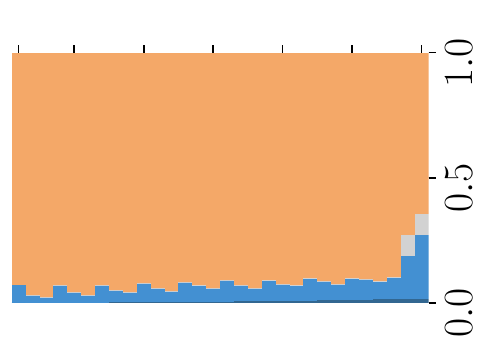}
\vspace{-0.32cm}\\
\rotatebox[origin=lt]{90}{\hspace{0.1cm}\small oracle $\theta_\text{DCG}$}\hspace{0.3cm} &
\includegraphics[scale=0.390]{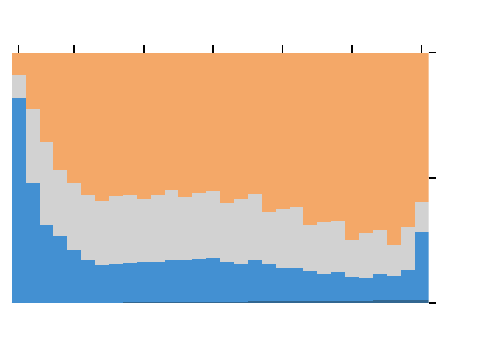} &
\includegraphics[scale=0.390]{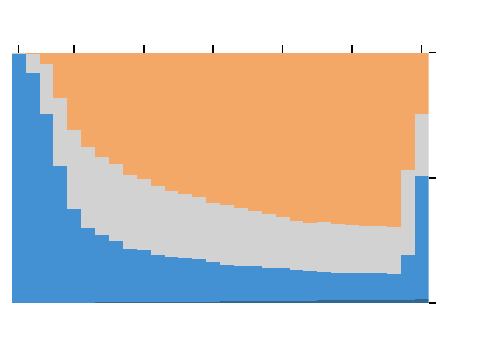} &
\includegraphics[scale=0.390]{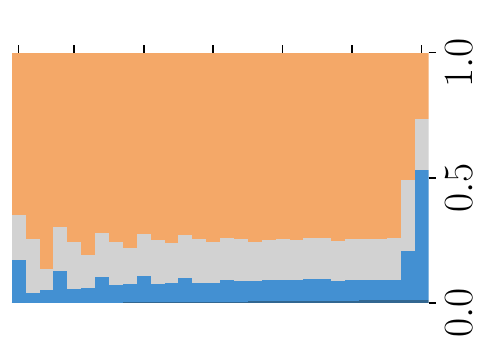}\vspace*{-0.25cm}\\
& \multicolumn{1}{c}{\rotatebox[origin=lt]{180}{\hspace{0.4cm}\small VLPL-2}}
&
 \multicolumn{1}{c}{\rotatebox[origin=lt]{180}{\hspace{0.4cm}\small slot-avg}}
 &
 \multicolumn{1}{c}{\rotatebox[origin=lt]{180}{\hspace{0.4cm}\small greedy}}\vspace{-0.3cm}%
\end{tabular}}
\begin{center}
\includegraphics[scale=0.45]{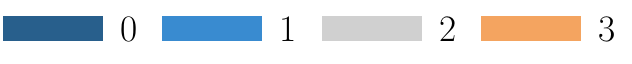}
\end{center}
\vspace{-0.3\baselineskip}
\vspace{-3.3mm}
\caption{Distribution of document lengths occupying each slot in learned rankings in MSLR. y-axis of stacked bar chart denotes the slot, x-axis denotes the proportion of rankings in the first test fold where the slot is occupied by the corresponding length. Length $0$ denotes padding.}
\label{fig:lengths}
\vspace{-0.5\baselineskip}
\end{figure}
As we've shown that VLPL can learn high reward rankings in the variable presentation length task, we then turn to~\ref{rq:rankings}: what document lengths does VLPL place in the learned rankings?

The lengths of documents occupying each slot for MSLR are shown in Figure~\ref{fig:lengths}. 
We first examine the behavior of VLPL-2 in the oracle setting, which corresponds to the task of finding the best possible ranking.
We see that VLPL-2 starts with a high probability of placing documents at length $1$ at the top of the ranking, followed by the increase in the proportion of lengths $2$ and $3$.
Length $3$ then occupies the vast majority of the ranking in the final positions.
Such behavior of showing lower-ranked documents at longer length was previously found to be beneficial by~\citet{marcos2015Effect}.
The shift from shorter to longer lengths is especially pronounced for $\theta_{\text{rank}^{-1}}$, suggesting that longer lengths and the associated increase in exposure may be particularly useful under stronger position bias.

Surprisingly, when document attractiveness is not known and instead estimated, VLPL follows a pattern that is roughly similar to the one in the oracle setting, yet with almost no use of length $2$.
We hypothesize that this may be due to the misestimation of $\rel$, which, as discussed in Subsection~\ref{sub:PRP} may lead to a different solution even if the relative ordering of all estimated $\rel(d,l)$ values is correct.

In contrast to VLPL, the slot-avg model places documents at a short length for longer, whilst the greedy model strongly prefers maximum length for most of the ranking, producing highly similar rankings for different $\theta$.
As the performance of VLPL-2 was previously shown to be significantly higher, the difference between models across different settings suggests that balancing document lengths is a critical factor that allows VLPL-2 to achieve strong performance, particularly in the oracle setting.

We saw that VLPL-2 strongly modulates its lengths. As such, for~\ref{rq:rankings} we conclude that: VLPL places shorter documents at the start of the ranking and progressively increases the placed document lengths, with a quicker transition for $\theta_{\text{rank}^{-1}}$.
{
\begin{figure}[tb]
\centering
\renewcommand{\arraystretch}{1}\begin{tabular}{c r r}
& \multicolumn{1}{c}{\hspace{0.45cm}\small MSLR}
& \multicolumn{1}{c}{\hspace{0.5cm}\small Yahoo}
\\
\vspace{-0.45cm}
\rotatebox[origin=lt]{90}{\hspace{-0.1cm} \small oracle \ac{EA}$_\text{DCG}$} &\hspace{-0.2cm}
\includegraphics[scale=0.41]{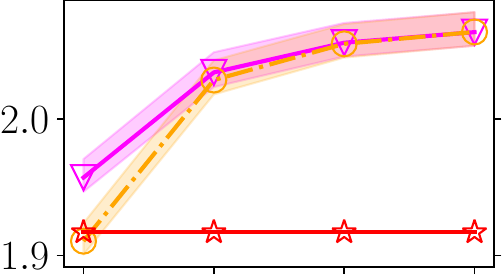} &
\raisebox{0.27mm}{\includegraphics[scale=0.41]{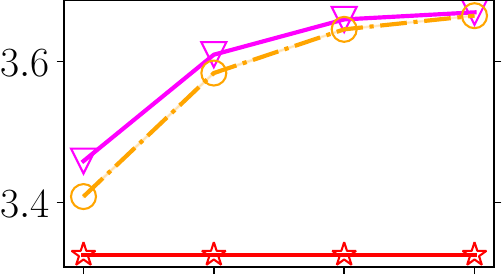}}%
\vspace*{0.35cm}%
\\
\rotatebox[origin=lt]{90}{\hspace{0.4cm} \small \ac{EA}$_\text{DCG}$}& 
\hspace{-0.2cm}\includegraphics[scale=0.41]{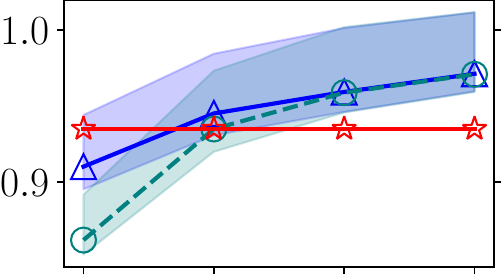} &
\includegraphics[scale=0.41]{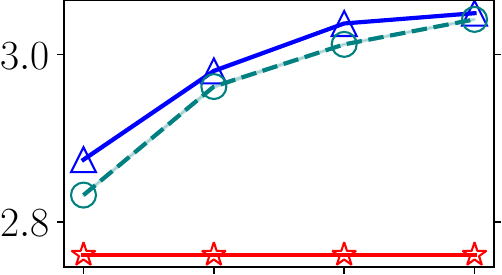}%
\vspace*{-0.1cm}
\\
\rotatebox[origin=lt]{90}{\hspace{-0.1cm} \small oracle \ac{EA}$_{\text{rank}^{-1}}$} &
\hspace{-0.2cm}\includegraphics[scale=0.41]{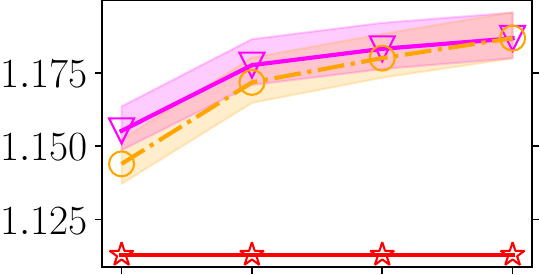} &
\includegraphics[scale=0.41]{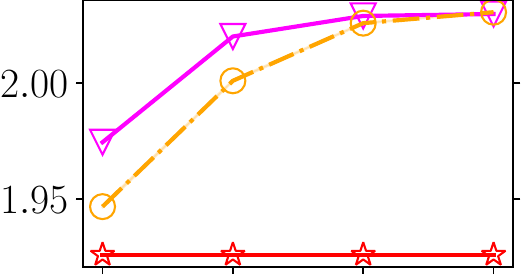}%
\vspace*{-0.05cm}
\\
\rotatebox[origin=lt]{90}{\hspace{0.4cm} \small \ac{EA}$_{\text{rank}^{-1}}$}& 
\hspace{-0.2cm}\includegraphics[scale=0.41]{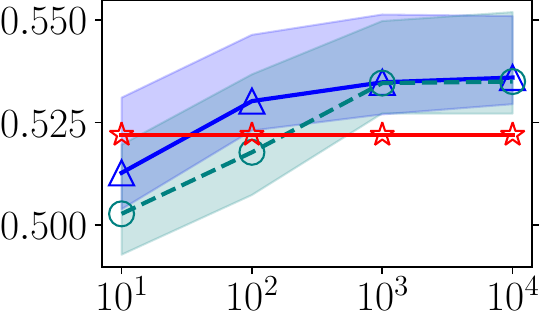} &
\includegraphics[scale=0.41]{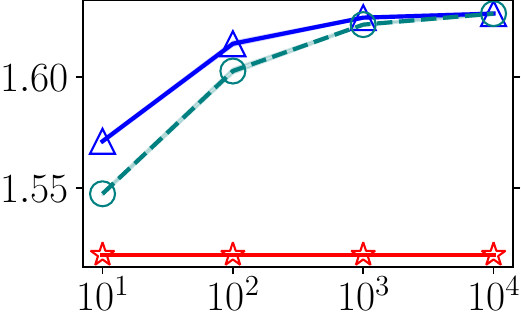}
\\
&\multicolumn{1}{c}{\small \hspace{1em} Number of Samples}
& \multicolumn{1}{c}{\small \hspace{1em} Number of Samples}\\
\multicolumn{3}{c}{
\includegraphics[scale=0.45]{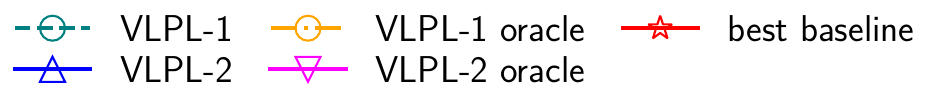}
}
\end{tabular}
\vspace{-0.5\baselineskip}
\vspace{-3.6mm}
\caption{\ac{EA} for post-processing VLPL models under different number of samples. Results are averages over five independent runs, the difference between the best and worst runs denoted by shading. Best baseline per setting denoted by star.}
\label{fig:N}
\vspace{-\baselineskip}
\end{figure}}

\subsection{Sample-efficiency}
\label{sub:sample_efficiency}

We then arrive at our final research question~(\ref{rq:N}): can VLPL models also achieve strong performance whilst using fewer samples?
To answer this, Figure~\ref{fig:N} shows VLPL performance under a varied number of samples $N$ in the oracle and non-oracle settings.

We can see that increasing the number of samples leads to improvements in \ac{EA}, with smaller, but still observable gains at larger sample numbers.
We also observe that whilst the performance of VLPL-1 and VLPL-2 is very similar at $N=10{,}000$, VLPL-2 shows stronger performance at lower values of $N$, with particularly substantial improvements over VLPL-1 at $N\leq 100$.
Nevertheless, for $N\geq 100$, both VLPL models generally achieve strong performance and always outperform the best baseline at $N=1{,}000$, highlighting the high applicability of VLPL to the variable presentation length ranking task.
As such, in response to~\ref{rq:N}, we affirm that: VLPL models, and especially VLPL-2, achieve large improvements over the baselines even when using fewer samples.
\section{Conclusion and Future Work}
\label{sec:conclusion}

In this paper we introduced the variable presentation length ranking setting, where the task is to simultaneously choose the document order and the size of each document's presentation in the ranking.
We showed that this setting is substantially more complex than standard ranking -- the best performance can only be achieved when document order and presentation lengths are decided jointly.
To tackle this problem, we introduced the variable document length Plackett-Luce model, and four VLPL methods for optimizing it.
Our experiments show 
that VLPL is highly suited for ranking with variable presentation lengths and 
underscore the importance of considering document presentation in \ac{LTR}.

Importantly, whilst we model document presentation lengths in terms of slots, it may also be possible to extend it to continuous values to represent pixel lengths of each document.
Furthermore, as different policies may be best suited for different presentation contexts, being able to learn a policy that is guaranteed to be effective across multiple settings (e.g., on mobile and desktop) may further help the application of such models.
Altogether, our work highlights the importance, the difficulty and the opportunities of \ac{LTR} with variable result presentation lengths.

\bigskip
\noindent
\textbf{Reproducibility}
Our experimental implementation is publicly available at \url{https://github.com/NKNY/varlenranksigir2025}.

\begin{acks}
This work used the Dutch national e-infrastructure with the support of the SURF Cooperative using grant no. EINF-9790. This work is partially supported by the Dutch Research Council (NWO), grant number VI.Veni.222.269.
\end{acks}

\clearpage
 \balance
\bibliographystyle{ACM-Reference-Format}
\bibliography{varlenrank2025}

\end{document}